\documentclass[10pt]{IEEEtran}
\usepackage[T1]{fontenc}
\usepackage[utf8]{inputenc}
\usepackage{xcolor}
\usepackage{float}
\usepackage{dsfont}
\usepackage{amsmath}
\usepackage{amsthm}
\usepackage{amssymb}
\usepackage{graphicx}
\usepackage[bookmarks=true,bookmarksnumbered=true,bookmarksopen=true,bookmarksopenlevel=1,
 breaklinks=false,pdfborder={0 0 0},pdfborderstyle={},backref=false,colorlinks=false]
 {hyperref}
\hypersetup{pdftitle={Your Title},
 pdfauthor={Your Name},
 pdfpagelayout=OneColumn, pdfnewwindow=true, pdfstartview=XYZ, plainpages=false}

\makeatletter
\ifx\proof\undefined\
  \newenvironment{proof}[1][\proofname]{\par
    \normalfont\topsep6\p@\@plus6\p@\relax
    \trivlist
    \itemindent\parindent
    \item[\hskip\labelsep
          \scshape
      #1]\ignorespaces
  }{%
    \endtrivlist\@endpefalse
  }
  \providecommand{\proofname}{Proof}
\fi

\usepackage[caption=false,font=footnotesize]{subfig}
\usepackage{tikz}
\usepackage{amsthm}

\usetikzlibrary{arrows}
\usetikzlibrary{shadows}

\usetikzlibrary{shapes}

\usetikzlibrary{arrows,decorations.pathmorphing}

\tikzset{
  every overlay node/.style={
    draw=white,anchor=north west,
  },
}
\usetikzlibrary{positioning, decorations.pathmorphing}

\usetikzlibrary{fit,positioning,hobby}

\usepackage{pdfsync}

\usepackage{algorithm,algpseudocode}

\usepackage{caption}
\makeatother

\theoremstyle{plain}
\newtheorem{lem}{\protect\lemmaname}
\theoremstyle{remark}
\newtheorem{rem}{\protect\remarkname}
\theoremstyle{definition}
\newtheorem{defn}{\protect\definitionname}
\newtheorem{example}{\protect\examplename}
\theoremstyle{plain}
\newtheorem{thm}{\protect\theoremname}
\newtheorem{cor}{\protect\corollaryname}
\providecommand{\corollaryname}{Corollary}
\providecommand{\definitionname}{Definition}
\providecommand{\examplename}{Example}
\providecommand{\lemmaname}{Lemma}
\providecommand{\remarkname}{Remark}
\providecommand{\theoremname}{Theorem}

\begin{document}
\title{Subspace Consensus}
\author{Lulu Pan,~\IEEEmembership{Member,~IEEE}, Yuhao Chen, Xiaohui Gong,
Haibin~Shao,~\IEEEmembership{Member,~IEEE} \thanks{The authors are with the School of Automation and Intelligent Sensing,
Shanghai Jiao Tong University, Shanghai 200240, China.}\thanks{A preliminary version of this work will be presented in 2027 IEEE
Conference on Decision and Control \cite{chen2027subspace}.}}
\maketitle
\begin{abstract}
This paper examines the subspace consensus problem for matrix-weighted
multi-agent networks. In this setting, each agent holds a high-dimensional
state vector in $\mathbb{R}^{d}$, and matrix-valued edge weights
characterize inter-agent interactions across different dimensions.
Unlike the classical scalar-weighted consensus problem, where all
dimensions of agents' states must converge to respective common values,
this paper examines scenarios in which agents must agree along certain
dimensions while preserving prescribed relative configurations in
the remaining ones. To address this need, we introduce the concept
of \textit{subspace consensus} for a matrix-weighted network, which
is said to achieve consensus on a subspace $\mathbb{V}\subseteq\mathbb{R}^{d}$
if the projections of all pairwise state differences onto $\mathbb{V}$
asymptotically vanish. This formulation subsumes the conventional
consensus problem as the special case $\mathbb{V}=\mathbb{R}^{d}$.
We first derive a purely algebraic necessary and sufficient condition
for subspace consensus by characterizing the interplay between the
null spaces of the matrix-valued Laplacian and those of the individual
matrix-valued edge weights. Then, we propose the \textsl{limiting
constraint network} as a means to analyze the steady state of matrix-weighted
networks. We establish the relationship between the steady state of
a matrix-weighted network and the null space of the matrix-valued
Laplacian of the associated limiting constraint network. Under mild
assumptions, we further show that a matrix-weighted network achieves
subspace consensus on $\mathbb{V}$ if and only if the associated
limiting constraint network contains a $\mathbb{V}$-spanning tree.
We further show that a matrix-weighted network achieves consensus if and only
if the associated limiting constraint network contains a positive
spanning tree, as a special case with $\mathbb{V}=\mathbb{R}^{d}$.
\end{abstract}

\begin{IEEEkeywords}
High-dimensional multi-agent systems, matrix-weighted networks, subspace
consensus, subspace spanning tree, limiting constraint network.
\end{IEEEkeywords}

\IEEEpeerreviewmaketitle{}

\section{Introduction}

The paradigm of distributed multi-agent coordination has become a
cornerstone in networked systems research, underpinning a broad range
of applications such as distributed estimation, control, optimization,
game, and learning over networks \cite{kia2019tutorial,mesbahi2010graph}.
At the heart of this paradigm lies the multi-agent consensus problem,
wherein a group of agents, each equipped with a local state, interact
through a communication network to asymptotically agree on a common
value \cite{olfati2004consensus,jadbabaie2002coordination}. The
classical consensus framework on scalar-weighted networks, however,
has long operated under the assumption that inter-agent interactions
are captured by scalar-valued edge weights, thereby overlooking the
potential complexity of interactions when agent states are vector-valued
\cite{friedkin2016network,sun2018dimensional,tian2025matrix}.

In many applications, the state of each agent naturally resides in
a higher-dimensional space. For instance, in formation control, an
agent's state may encode its position and orientation \cite{zhao2015translational};
in opinion dynamics, it may represent stances on multiple topics \cite{ye2020continuous};
in sensor networks, it may capture multidimensional measurements \cite{Barooah2007Estimation}.
In such settings, interactions between agents are not merely scalar-weighted
but can involve intricate couplings across different dimensions of
the state vectors. This observation has motivated a growing body of
research on matrix-weighted networks, where each edge is endowed with
a matrix-valued weight that modulates the influence between agents
across the dimensions of their states \cite{trinh2018matrix,li2020synchronization,pan2018bipartite,tuna2017observability,pan2025privacy,wang2022characterizing,miao2022matrix,pan2022cluster}.
Such matrix-weighted couplings naturally arise in scenarios such as
generalized effective resistance in electrical networks \cite{barooah2008estimation},
multi-topic opinion dynamics \cite{ye2020continuous}, bearing-based
distributed formation \cite{zhao2015translational}, and the dynamics
of arrays of coupled oscillators \cite{tuna2019synchronization}.

However, most existing results on matrix-weighted networks focus on
a particular type of collective behavior, namely full-dimensional
consensus, in which all components of the agents’ state vectors converge
to a common value. These formulations, while important, do not fully
capture the richness of behaviors that matrix-weighted interactions
can engender. In particular, there are many scenarios in which full-dimensional
consensus is either unnecessary or even undesirable. Consider, for
example, a multi-agent system tasked with achieving a specific formation
in a subset of dimensions while maintaining agreement in the remaining
ones. In bearing-based formation control, agents often use orthogonal
projection matrices as edge weights, thereby restricting interactions
to a subspace. The protocol then drives the state differences only
within that subspace, leaving components orthogonal to it unaffected
\cite{zhao2015translational,wei2023designing}. Such systems do not
converge to a single common vector; instead, they reach a configuration
in which some dimensions are aligned while others maintain prescribed
relative offsets. This kind of behavior lies beyond the reach of the
traditional consensus analysis paradigm.

Remarkably, the matrix-weighted coupling mechanism inherently supports
such dimension-specific agreement patterns. The protocol drives to
zero only the projection of the state difference onto the row space
of the weight matrix, leaving components in its null space untouched.
This fundamental property suggests that by carefully designing the
edge weight matrices (in particular, their row spaces), one can engineer
systems in which agreement is enforced only along prescribed directions,
while allowing flexibility or even prescribed structures in the orthogonal
complement. Yet, despite its conceptual appeal and practical relevance,
this inherent capability has remained largely unnoticed in the literature.

\textit{Contribution Statement:} This paper aims to fill this gap
by introducing the concept of \textit{subspace consensus}. A matrix-weighted
network is said to achieve subspace consensus on $\mathbb{V}\subseteq\mathbb{R}^{d}$
if, for every pair of agents, the projection of their state difference
onto $\mathbb{V}$ asymptotically vanishes. We provide a comprehensive
analysis of the subspace consensus problem from both algebraic and
topological perspectives. From an algebraic standpoint, we derive
necessary and sufficient conditions for subspace consensus by examining
the structure of the matrix-weighted Laplacian and its null space.
We show that subspace consensus on $\mathbb{V}$ is intimately related
to the alignment of the row spaces of edge weights along paths in
the network. From a topological perspective, we propose the limiting
constraint network (constructed iteratively from constraints on pairwise
state differences) as a means to analyze the steady state of matrix-weighted
networks. We establish the relationship between the steady state of
a matrix-weighted network and the null space of the matrix-valued
Laplacian of the associated limiting constraint network. Under mild
assumptions, we further show that a multi-agent system achieves subspace
consensus on $\mathbb{V}$ if and only if the associated limiting
constraint network contains a $\mathbb{V}$-spanning tree. Finally,
we show that a matrix-weighted network achieves consensus if and only
if the associated limiting constraint network contains a positive
spanning tree, as a special case of subspace consensus with $\mathbb{V}=\mathbb{R}^{d}$.

The remainder of the paper is organized as follows. $\mathsection$\ref{sec:problem-formulation}
and $\mathsection$\ref{sec:Motivation-and-Problem} introduce the
preliminaries and motivate the subspace consensus problem, respectively.
Main theoretical and numerical results are presented in $\mathsection$\ref{sec:Subspace-Consensus}.
We finally provide concluding remarks in $\mathsection$\ref{sec:Conclusion}.

\section{Notation and Preliminaries\label{sec:problem-formulation}}

\subsection{Notations}

We first introduce the notation. Let $\mathbb{R}$, $\mathbb{N}$,
and $\mathbb{Z}_{+}$ be the set of real numbers, natural numbers,
and positive integers, respectively. For $n\in\mathbb{Z}_{+}$, denote
$\underline{n}=\left\{ 1,2,\ldots,n\right\} $. We use $M\succ0$
(respectively, $M\succeq0$) to denote that a symmetric matrix $M$
is positive definite (respectively, positive semi-definite).\textcolor{olive}{{}
}The null space, row space, and range space of a matrix $M$ are denoted
by $\text{{\bf null}}(M)$, $\text{{\bf row}}(M)$ and $\text{{\bf range}}(M)$,
respectively. Let $\mathds{1}_{n}$, $0_{n\times n}$ and $I_{n}$
designate the $n$-dimensional column vector whose components are
all $1$'s, the $n\times n$ matrix whose components are all $0$'s,
and the $n\times n$ identity matrix, respectively. The $i$-th entry
of a vector $\boldsymbol{x}\in\mathbb{R}^{d}$ is denoted by $[\boldsymbol{x}]_{i}$.
A non-empty subset $\mathbb{V}\subseteq\mathbb{R}^{d}\thinspace(d\in\mathbb{Z}_{+})$
is a subspace if it is closed under addition and scalar multiplication.
For a subspace $\mathbb{V}\subseteq\mathbb{R}^{d}$, let $\text{{\bf basis}}(\mathbb{V})$
denote a basis of $\mathbb{V}$. For a subspace $\mathbb{V}\subseteq\mathbb{R}^{d}$,
its orthogonal complement is denoted by $\mathbb{V}^{\perp}$. For
two subspaces $\mathbb{U}$ and $\mathbb{V}$, their intersection
and sum are denoted by $\mathbb{U}\cap\mathbb{V}$ and $\mathbb{U}+\mathbb{V}$,
respectively. For two orthogonal subspaces $\mathbb{U}$ and $\mathbb{V}$,
we write $\mathbb{U}\perp\mathbb{V}$. For a subspace $\mathbb{V}\subseteq\mathbb{R}^{d}$,
the matrix $\text{P}_{\mathbb{V}}:\mathbb{R}^{d}\to\mathbb{V}$ denotes
the orthogonal projection onto $\mathbb{V}$. Let $\boldsymbol{e}_{i}\in\mathbb{R}^{d}$
denote the vector whose $i$-th entry is $1$, and all other entries
are $0$. We use $\mathbf{span}\left\{ \boldsymbol{x}_{1},\cdots,\boldsymbol{x}_{s}\right\} $
to represent the vector space linearly spanned by vectors $\boldsymbol{x}_{1},\cdots,\boldsymbol{x}_{s}$.
 The symmetric difference of two sets $A\subseteq\mathbb{R}^{d}$
and $B\subseteq\mathbb{R}^{d}$ is $A\triangle B:=(A\setminus B)\cup(B\setminus A)$,
where $A\setminus B=\{\boldsymbol{x}\mid\boldsymbol{x}\in A\text{ and }\boldsymbol{x}\notin B\}$
and $B\setminus A=\{\boldsymbol{x}\mid\boldsymbol{x}\in B\text{ and }\boldsymbol{x}\notin A\}$.

\subsection{Matrix-Weighted Networks}

Let $\mathcal{G}=(\mathcal{V},\mathcal{E},A)$ be a graph on $n>1$
($n\in\mathbb{Z}_{+}$) nodes. The node and edge sets of $\mathcal{G}$
are denoted by $\mathcal{V}=\left\{ 1,2,\ldots,n\right\} $ and $\mathcal{E}\subseteq\mathcal{V}\times\mathcal{V}$,
respectively. A matrix-weighted multi-agent network (or simply matrix-weighted
network) is a graph where each edge $(i,j)\in\mathcal{E}$ is associated
with a matrix-valued weight $A_{ij}\in\mathbb{R}^{d\times d}$ such
that $A_{ij}\not=0_{d\times d}$ if $(i,j)\in\mathcal{E}$, and $A_{ij}=0_{d\times d}$
otherwise. Unlike scalar-weighted networks, the matrix-valued adjacency
matrix $A=(A_{ij})\in\mathbb{R}^{dn\times dn}$ is a block matrix
whose block located in its $i$-th row and the $j$-th column is $A_{ij}$
for all $i,j\in\underline{n}$.\textcolor{red}{{} }The matrix-valued
Laplacian matrix of $\mathcal{G}$ is defined as $L=D-A$ where $D=\text{{\bf diag}}\left\{ D_{1},\cdots,D_{n}\right\} \in\mathbb{R}^{dn\times dn}$
and $D_{i}=\underset{j\in\mathcal{N}_{i}}{\sum}A_{ij}\in\mathbb{R}^{d\times d}$.
The matrix-valued Laplacian matrix is defined as $L=D-A\in\mathbb{R}^{dn\times dn}$.
Let $\mathcal{N}_{i}=\left\{ j\in\mathcal{V}\,|\,(i,j)\in\mathcal{E}\right\} $
denote the neighbor set of agent $i$. For nodes $i\in\mathcal{V}$
and $j\in\mathcal{V}$, let $q^{ij}=(v_{0},\cdots,v_{k})$ denote
a path from $i$ to $j$ where $v_{0}=i,v_{k}=j$ and $(v_{s},v_{s+1})\in\mathcal{E},\:\forall s=0,\cdots,k-1.$
Let $q^{ij}_{k}$ denote the $k$-th path from node $i$ to node $j$
where $k\in\underline{m_{ij}}$. Let $\mathcal{Q}_{ij}=\left\{ q^{ij}_{k}\;|\;k\in\underline{m_{ij}}\right\} $
denote the set of all paths from node $i$ to node $j.$ A matrix-weighted
network is connected if there is a path between every pair of distinct
nodes. An edge $(i,j)\in\mathcal{E}$ in $\mathcal{G}$ is positive
definite (respectively, positive semi-definite) if $A_{ij}\succ0$
(respectively, $A_{ij}\succeq0$). A tree $\mathcal{T}$ in $\mathcal{G}$
is a $\mathbb{V}$-tree if $\mathbb{V}\subseteq\text{{\bf row}}(A_{ij})$
for all $(i,j)\in\mathcal{E}(\mathcal{T}).$ A $\mathbb{V}$-spanning
tree $\mathcal{T}$ in $\mathcal{G}$ is a $\mathbb{V}$-tree containing
all nodes of $\mathcal{G}$. A tree $\mathcal{T}$ in $\mathcal{G}$
is a positive tree if every edge contained in this tree is positive
definite. A positive spanning tree $\mathcal{T}$ in $\mathcal{G}$
is a positive tree containing all nodes of $\mathcal{G}$. 

\section{Motivation and Problem Formulation\label{sec:Motivation-and-Problem}}

Consider a matrix-weighted network $\mathcal{G}=(\mathcal{V},\mathcal{E},A)$
consisting of $n$ agents. The state of each agent $i\in\mathcal{V}$
is denoted by $\boldsymbol{x}_{i}\in\mathbb{R}^{d}\thinspace(d\in\mathbb{Z}_{+})$
which evolves according to the interaction protocol
\begin{equation}
\dot{\boldsymbol{x}}_{i}(t)={\displaystyle \sum_{j\in\mathcal{N}_{i}}}A_{ij}\left(\boldsymbol{x}_{j}(t)-\boldsymbol{x}_{i}(t)\right),i\in\mathcal{V}.\label{eq:protocol}
\end{equation}
The overall dynamics of \eqref{eq:protocol} is dictated by 
\begin{equation}
\dot{\boldsymbol{x}}(t)=-L\boldsymbol{x}(t)\label{eq:overall}
\end{equation}
where $\boldsymbol{x}(t)=(\boldsymbol{x}^{\top}_{1}(t),\boldsymbol{x}^{\top}_{2}(t),\ldots\boldsymbol{x}^{\top}_{n}(t))^{\top}$.

Note that the protocol \eqref{eq:protocol} extends the well-known
scalar-weighted consensus protocol by introducing matrix-valued edge
weights $A_{ij}$ \cite{olfati2004consensus,tuna2016synchronization,trinh2018matrix,pan2018bipartite}.
Therefore, the null space of the matrix-weighted Laplacian $L$ in
\eqref{eq:overall} expands, and the collective behavior of \eqref{eq:overall}
becomes intricate \cite{trinh2018matrix,pan2018bipartite,wang2022characterizing}. 
\begin{lem}
\cite{trinh2018matrix,pan2021bipartite,pan2018bipartite} \label{lem:null-space}
The matrix-valued Laplacian matrix $L$ of a matrix-weighted network
$\mathcal{G}=(\mathcal{V},\mathcal{E},A)$ is positive semi-definite,
and the null space of $L$ can be characterized by $\text{{\bf null}}(L)=\text{{\bf span}}\left\{ \mathcal{R},\mathcal{H}\right\} ,$
where 
\begin{equation}
\mathcal{R}=\text{{\bf range}}\{\mathds{1}_{n}\otimes I_{d}\},\label{eq:consensus-space}
\end{equation}
and 
\begin{align}
\mathcal{H}=\{\boldsymbol{v} & =(\boldsymbol{v}^{\top}_{1},\boldsymbol{v}^{\top}_{2},\cdots,\boldsymbol{v}^{\top}_{n})^{\top}\in\mathbb{R}^{dn}\mid\nonumber \\
 & (\boldsymbol{v}_{i}-\boldsymbol{v}_{j})\in\text{{\bf null}}(A_{ij}),\,(i,j)\in\mathcal{E}\}.\label{eq:edge-space}
\end{align}
\end{lem}
Note from Lemma \ref{lem:null-space} that the null space of the matrix-valued
edge weights $A_{ij}$ inevitably influences the null space of the
matrix-valued Laplacian matrix. If matrix-valued edge weights in $\mathcal{G}$
are all positive definite, then $\text{{\bf null}}(L)=\mathcal{R}$,
which implies all agents' states will asymptotically converge to the
same value. However, if there exist positive semi-definite edges in
$\mathcal{G}$, then the null space of $L$ can be enlarged beyond
$\mathcal{R}$ due to the nontrivial null space of the positive semi-definite
edges. This fact motivates us to examine the algebraic properties
of $A_{ij}$ and their influence on the steady state of the matrix-weighted
system \eqref{eq:overall}. Specifically, we shall \textcolor{black}{first}
examine the influence of a single matrix-valued weight matrix on the
state difference of neighboring agents in protocol \eqref{eq:protocol}.
One can decompose the relative state difference $\Delta\boldsymbol{x}_{ij}=\boldsymbol{x}_{j}-\boldsymbol{x}_{i}$
into two orthogonal components with respect to the null space and
row space of the matrix $A_{ij}$. We can write 
\begin{equation}
\Delta\boldsymbol{x}_{ij}=\text{P}_{\text{{\bf null}}(A_{ij})}(\Delta\boldsymbol{x}_{ij})+\text{P}_{\text{{\bf row}}(A_{ij})}(\Delta\boldsymbol{x}_{ij}).
\end{equation}
Note that any vector in $\text{{\bf null}}(A_{ij})$ is annihilated
by $A_{ij},$ the term $A_{ij}\text{P}_{\text{{\bf null}}(A_{ij})}(\Delta\boldsymbol{x}_{ij})$
equals zero. Consequently, 
\begin{equation}
A_{ij}\Delta\boldsymbol{x}_{ij}=A_{ij}\mathbf{\text{P}}_{\text{{\bf row}}(A_{ij})}(\Delta\boldsymbol{x}_{ij}).
\end{equation}
This implies that the interaction along edge $(i,j)$ is insensitive
to the component of $\Delta\boldsymbol{x}_{ij}$ lying in $\text{{\bf null}}(A_{ij}).$
In contrast, the component projected onto $\text{{\bf row}}(A_{ij})$
is actively affected by the dynamics, and the consensus protocol will
gradually drive this row-space component to zero over time.
\begin{rem}
For scalar-weighted networks, $A_{ij}=a_{ij}I_{d},$ which implies
that $\text{{\bf row}}(A_{ij})=\mathbb{R}^{d}$ and $\text{{\bf null}}(A_{ij})=\left\{ 0\right\} $.
Therefore, equation \eqref{eq:protocol} drives the projection of
$\Delta\boldsymbol{x}_{ij}$ onto $\text{{\bf row}}(A_{ij})$, i.e.,
$\Delta\boldsymbol{x}_{ij}$, to gradually decay to zero. This implies
that any two adjacent nodes converge to the same value, thereby achieving
consensus. 

From a theoretical perspective, the classical consensus protocol is
defined in the whole state space $\mathbb{R}^{d}$, where all agents
asymptotically agree on a common vector in $\mathbb{R}^{d}$. This
naturally raises the following question: what happens if agreement
is required only on a subspace of $\mathbb{R}^{d}$? From a practical
perspective, many matrix-weighted networks inherently exhibit such
a structure. For instance, bearing-based algorithms typically use
orthogonal projection matrices as matrix-valued edge weights, thereby
restricting interactions to specific subspaces \cite{zhao2015translational}.
In these cases, the protocol only regulates state differences within
the projected subspace, while components in its orthogonal complement
remain unaffected. This observation motivates the concept of subspace
consensus, in which agents are required to asymptotically agree only
on a prescribed subspace rather than the entire space.
\end{rem}
\begin{defn}
\textcolor{violet}{\label{def:subspace consensus} }\textcolor{black}{Let
$\mathbb{V}$ be a subspace of $\mathbb{R}^{d}$. The} matrix-weighted
network \eqref{eq:overall} a\textcolor{black}{chieves subspace consensus
on $\mathbb{V}$ if 
\begin{equation}
\lim_{t\rightarrow\infty}\text{P}_{\mathbb{V}}\left(\boldsymbol{x}_{i}(t)-\boldsymbol{x}_{j}(t)\right)=\boldsymbol{0},\thinspace\forall i,j\in\mathcal{V},\label{eq:subspace consensus}
\end{equation}
for }an arbitrary initial condition\textcolor{black}{{} $\boldsymbol{x}(0)\in\mathbb{R}^{dn}$}.
\end{defn}
\textcolor{black}{}
\textcolor{black}{According to Definition \ref{def:subspace consensus},
subspace consensus requires that the projections of all agents' states
onto the subspace $\mathbb{V}$ asymptotically coincide, while the
orthogonal components in $\mathbb{V}^{\perp}$ may vary. To see this,
we provide the following Example.}
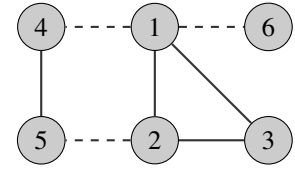
\begin{figure}[tbh]
\centering{}\begin{tikzpicture}[scale=1]
    \node (n4) at (0,1.5) [circle,fill=black!20,draw] {4};
    \node (n5) at (0,0) [circle,fill=black!20,draw] {5};
    \node (n2) at (1.5,0) [circle,fill=black!20,draw] {2};
	\node (n1) at (1.5,1.5) [circle,fill=black!20,draw] {1};
    \node (n3) at (3,0) [circle,fill=black!20,draw] {3};
	\node (n6) at (3,1.5) [circle,fill=black!20,draw] {6};

	\draw [-,  thick, color=black!80] (n1) -- (n2); 
	\draw [-,  thick, color=black!80] (n1) -- (n3); 
    \draw [-,  thick, color=black!80] (n2) -- (n3); 

	\draw [-,  thick, color=black!80] (n5) -- (n4); 

    \draw [-,dashed,  thick, color=black!80] (n1) -- (n4); 
    \draw [-,dashed,  thick, color=black!80] (n2) -- (n5); 

    \draw [-,dashed,  thick, color=black!80] (n1) -- (n6); 
\end{tikzpicture}\caption{\label{fig:The-network-structure of exm for def}A $6$-node matrix-weighted
network. Dashed lines and solid lines represent positive semi-definite
edges and positive definite edges, respectively.}
\end{figure}
\textcolor{red}{{} }
\begin{figure}[H]
\begin{centering}
\textcolor{black}{\includegraphics[width=10cm]{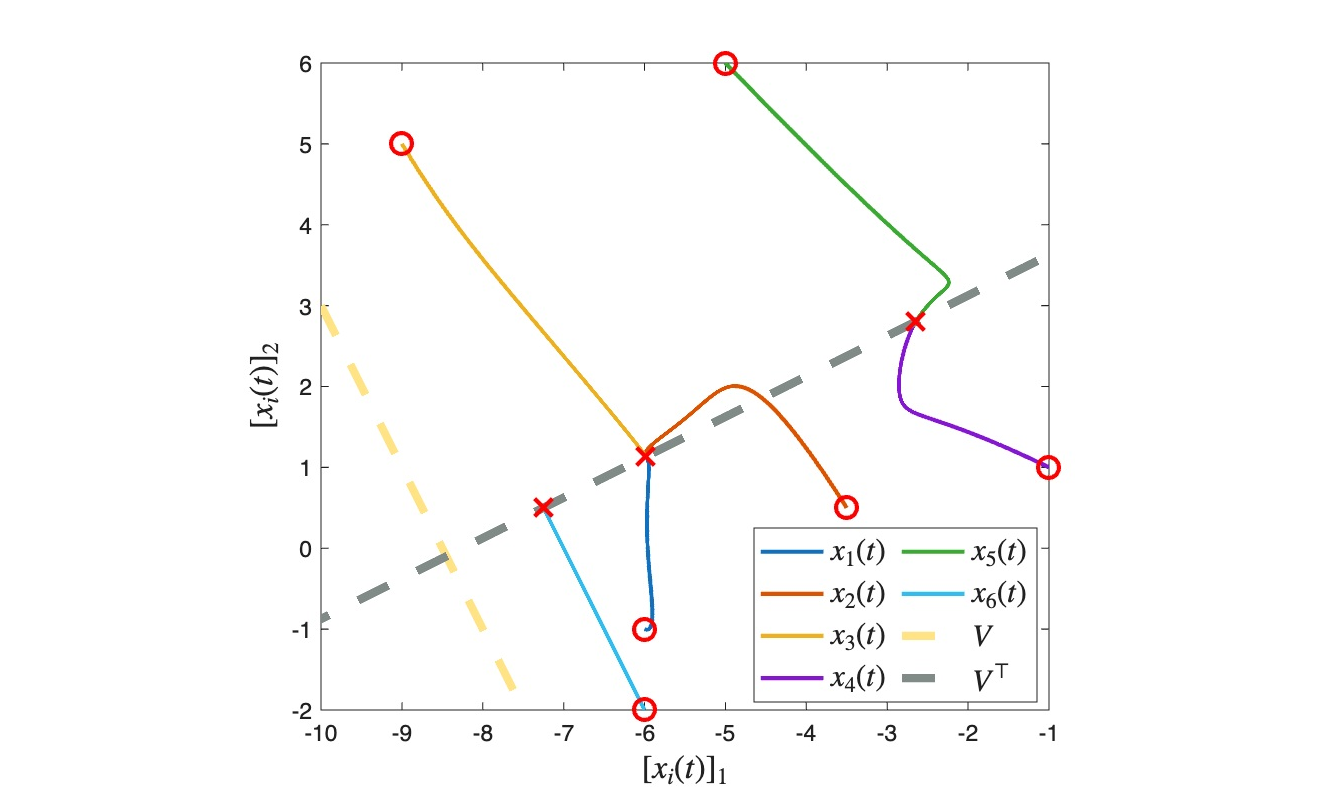}}
\par\end{centering}
\centering{}\textcolor{black}{\caption{\label{fig:Example_for_definition}The trajectories of agents' states.
The initial state and the steady state of each agent are indicated
by red triangle and red star, respectively. The red dashed and the
blue dashed lines represent the subspace $\mathbb{V}$ and its orthogonal
complement \textcolor{black}{$\mathbb{V}^{\perp}$}, respectively.}
}
\end{figure}

\begin{example}
\textcolor{black}{\label{exa1}}\textcolor{blue}{} Consider the matrix-weighted
network \eqref{eq:overall} with $\mathcal{G}$ shown in Figure \ref{fig:The-network-structure of exm for def},
where the state of each agent is two-dimensional. Let $\mathbb{V}=\text{{\bf span}}\left\{ \boldsymbol{e}_{1}-2\boldsymbol{e}_{2}\right\} $,
the row spaces of the weight matrices are configured as
\[
\text{{\bf row}}(A_{12})=\text{{\bf row}}(A_{13})=\text{{\bf row}}(A_{23})=\text{{\bf row}}(A_{45})=\mathbb{R}^{2},
\]
 and 
\[
\text{{\bf row}}(A_{14})=\text{{\bf row}}(A_{16})=\text{{\bf row}}(A_{25})=\mathbb{V}.
\]
\textcolor{black}{The trajectories of agents' states are shown in
Figure \ref{fig:Example_for_definition}, where the projections of
all agents' states onto the subspace $\mathbb{V}$ asymptotically
coincide, while their orthogonal components in $\mathbb{V}^{\perp}$
are different. It is shown that subspace consensus on $\mathbb{V}$
is achieved. }
\end{example}
\textcolor{black}{}Definition \ref{def:subspace consensus} provides
a unified lens through which to view and analyze a broad spectrum
of collective behaviors. Notably, subspace consensus naturally extends
the classical notion of consensus. When $\mathbb{V}=\mathbb{R}^{d}$,
the definition reduces to the familiar requirement that all agents
converge to a common vector in the full space \textcolor{black}{\cite{olfati2004consensus}}.

The introduction of subspace consensus opens up a rich set of theoretical
questions. Under what conditions on the network topology and the matrix-valued
edge weights does a given network achieve subspace consensus on a
prescribed subspace? How does the structure of the weight matrices
interact with the graph connectivity to determine the asymptotic behavior?
Addressing these questions requires a departure from the scalar-weighted
intuition; in matrix-weighted networks, the properties of the weight
matrices, in particular their null and row spaces, are equally decisive.

\section{Main Results\label{sec:Subspace-Consensus}\protect 
}In this section, we investigate the conditions under which subspace
consensus can be achieved. We begin our analysis by presenting an
algebraic condition for the matrix-weighted network. As stated in
Theorem \ref{thm:algebraic condition}, the achievement of subspace
consensus is equivalent to a specific property of the null space of
the matrix-valued Laplacian.
\begin{thm}
\label{thm:algebraic condition} The matrix-weighted network \eqref{eq:overall}
achieves subspace consensus on $\mathbb{V}\subseteq\mathbb{R}^{d}$
if and only if 
\begin{align*}
\text{{\bf null}}(L)=\{\boldsymbol{v} & =(\boldsymbol{v}^{\top}_{1},\boldsymbol{v}^{\top}_{2},\cdots,\boldsymbol{v}^{\top}_{n})^{\top}\in\mathbb{R}^{dn}\mid\\
 & (\boldsymbol{v}_{i}-\boldsymbol{v}_{j})\in\text{{\bf null}}(A_{ij})\cap\text{{\bf null}}(\mathrm{P}_{\mathbb{V}}),\,\forall(i,j)\in\mathcal{E}\}.
\end{align*}
\end{thm}
\begin{proof}
(Sufficiency).\textcolor{black}{{} Note that the matrix-valued Laplacian
$L$ is positive semidefinite. Applying the orthogonal decomposition
yields $L=P\Lambda P^{\top}$, where $P=\left(\boldsymbol{p}_{1},\boldsymbol{p}_{2},\cdots,\boldsymbol{p}_{nd}\right)$
is an orthogonal matrix whose columns are eigenvectors of $L$, and
$\Lambda=\operatorname{diag}\{\lambda_{1},\lambda_{2},\dots,\lambda_{nd}\}$
satisfies $\lambda_{1}=\dots=\lambda_{k}=0<\lambda_{k+1}\le\dots\le\lambda_{nd}$.
The converged value of the system \eqref{eq:overall} satisfies
\begin{eqnarray}
\boldsymbol{x}^{*} & = & \lim_{t\rightarrow\infty}\boldsymbol{x}(t)=\lim_{t\rightarrow\infty}e^{-Lt}\boldsymbol{x}(0)\\
 & = & \lim_{t\rightarrow\infty}Pe^{-\Lambda t}P^{\top}\boldsymbol{x}(0)\\
 & = & \lim_{t\rightarrow\infty}\sum^{nd}_{i=1}e^{-\lambda_{i}t}\boldsymbol{p}_{i}\boldsymbol{p}^{\top}_{i}\boldsymbol{x}(0)\\
 & = & \sum^{k}_{i=1}\left(\boldsymbol{p}^{\top}_{i}\boldsymbol{x}(0)\right)\boldsymbol{p}_{i},
\end{eqnarray}
where $\boldsymbol{x}^{*}=[\boldsymbol{x}^{*}_{1}{}^{\top},\boldsymbol{x}^{*}_{2}{}^{\top},...,\boldsymbol{x}^{*}_{n}{}^{\top}]^{\top}$,
$\boldsymbol{x}^{*}_{q}\in\mathbb{R}^{d}$ and $q\in\underline{n}$.
Hence, for any $(l,m)\in\mathcal{E}$, one has,
\begin{eqnarray}
\boldsymbol{x}^{*}_{l}-\boldsymbol{x}^{*}_{m} & = & \sum^{k}_{i=1}\left(\boldsymbol{p}^{\top}_{i}\boldsymbol{x}(0)\right)\left(\boldsymbol{p}_{il}-\boldsymbol{p}_{im}\right),
\end{eqnarray}
since $\boldsymbol{p}_{il}-\boldsymbol{p}_{im}\in\text{{\bf null}}(A_{lm})\cap\text{{\bf null}}(\mathrm{P}_{\mathbb{V}})$,
therefore one has,
\begin{equation}
(\boldsymbol{x}^{*}_{l}-\boldsymbol{x}^{*}_{m})\in\text{{\bf null}}(A_{lm})\cap\text{{\bf null}}(\mathrm{P}_{\mathbb{V}}).
\end{equation}
Therefore, for any $\boldsymbol{x}(0)\in\mathbb{R}^{dn}$, one has,
\begin{equation}
\lim_{t\rightarrow\infty}\text{P}_{\mathbb{V}}\left(\boldsymbol{x}_{l}(t)-\boldsymbol{x}_{m}(t)\right)=\mathrm{P}_{\mathbb{V}}(\boldsymbol{x}^{*}_{l}-\boldsymbol{x}^{*}_{m})=0,\forall(l,m)\in\mathcal{E},
\end{equation}
since the matrix-weighted network is connected, which is equivalent
to the }matrix-weighted network\textcolor{black}{{} \eqref{eq:overall}
achieves subspace consensus on $\mathbb{V}$.}

\noindent\textcolor{black}{(Necessity). Assume that the }matrix-weighted
network\textcolor{black}{{} \eqref{eq:overall} achieves subspace consensus
on $\mathbb{V}$ and there exists a vector $\boldsymbol{v}^{*}=(\boldsymbol{v}^{\top}_{1},\boldsymbol{v}^{\top}_{2},\cdots,\boldsymbol{v}^{\top}_{n})^{\top}\in\mathbb{R}^{dn}\in\text{{\bf null}}(L)$
such that $\boldsymbol{v}^{*}\in\text{{\bf null}}(L)$ but $\boldsymbol{v}^{*}\notin\text{{\bf null}}(\mathrm{P}_{\mathbb{V}})$.
Then, we can choose an initial state$\boldsymbol{x}(0)=\boldsymbol{v}^{*}$
and has $\underset{t\rightarrow\infty}{\lim}\boldsymbol{x}(t)=\boldsymbol{v}^{*}$.
Since $\boldsymbol{v}^{*}\notin\text{{\bf null}}(\mathrm{P}_{\mathbb{V}})$,
therefore, the }matrix-weighted network\textcolor{black}{{} \eqref{eq:overall}
cannot achieve the subspace consensus on $\mathbb{V}$, the proof
is finished.}
\end{proof}
\begin{rem}
Theorem \ref{thm:algebraic condition} eventually explicitly characterizes
the structure of $\text{{\bf null}}(L)$ in the context of the subspace
consensus framework. By Theorem \ref{thm:algebraic condition}, one
can revisit the well-known necessary and sufficient condition for
the classical consensus problem \cite{olfati2004consensus}. In the
scalar-weighted network case, one only needs to choose $\mathbb{V}=\mathbb{R}^{d}$
and $\mathrm{P}_{\mathbb{V}}=I_{d}$, then the matrix-weighted network
\eqref{eq:overall} achieves consensus on $\mathbb{R}^{d}$ if and
only if
\begin{align*}
\text{{\bf null}}(L)=\{\boldsymbol{v} & =(v_{1},v_{2},\cdots,v_{n})^{\top}\in\mathbb{R}^{n}\mid\\
 & (v_{i}-v_{j})=0,\,\forall(i,j)\in\mathcal{E}\}.
\end{align*}
\end{rem}
Although Theorem \ref{thm:algebraic condition} characterizes subspace
consensus via the null space of the matrix-valued Laplacian, an alternative
perspective can be gained by examining the row spaces of the edge
weights. As the steady state differences must lie in the $\mathbf{null}(A_{ij})$,
they are necessarily orthogonal to $\mathbf{row}\left(A_{ij}\right)$.
Consequently, a simple and easily verifiable choice is to take consensus
subspace $\mathbb{V}$ to belong to the row spaces of all edge weights,
i.e., $\mathbb{V\subseteq}\mathbf{row}\left(A_{ij}\right)$ for all
$(i,j)\in\mathcal{E}$; in particular, the intersection of these row
spaces provides the most straightforward candidate for identifying
consensus subspace $\mathbb{V}$. Nevertheless, this local intersection
approach may fail to capture the true global behavior of the network.
In the following example, we show that even if the intersection of
all edge row spaces is trivial, subspace consensus can still be achieved
due to the underlying global topological structure.
\begin{figure}[tbh]
\centering{}\begin{tikzpicture}[scale=0.6]
    \node (n2) at (0.5,3) [circle,fill=black!20,draw] {2};
    \node (n1) at (0.5,0) [circle,fill=black!20,draw] {1};
    \node (n4) at (3.5,0) [circle,fill=black!20,draw] {4};
	\node (n3) at (3.5,3) [circle,fill=black!20,draw] {3};

	\draw [-,  dashed, thick, color=black!80] (n1) -- (n2); 
	\draw [-,  dashed, thick, color=black!80] (n1) -- (n3); 
	\draw [-,  dashed, thick, color=black!80] (n1) -- (n4); 
    \draw [-,  dashed, thick, color=black!80] (n2) -- (n3); 
	\draw [-,  dashed, thick, color=black!80] (n2) -- (n4); 
    \draw [-,  dashed, thick, color=black!80] (n3) -- (n4); 
\end{tikzpicture}\caption{\label{The 4-node network}A 4-node matrix-weighted network.}
\end{figure}
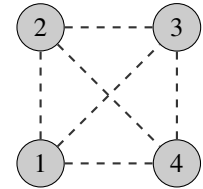

\begin{example}
\label{iterative algorithm-motivation example}Consider a matrix-weighted
network $\mathcal{G}$ as shown in Figure \ref{The 4-node network},
where the state of each agent is three-dimensional. The null spaces
of the matrix-valued edge weights are given as follows: $\text{{\bf null}}(A_{12})=\text{{\bf span}}\left\{ \boldsymbol{e}_{1}+\boldsymbol{e}_{2}+2\boldsymbol{e}_{3}\right\} ,$
$\text{{\bf null}}(A_{13})=\text{{\bf span}}\left\{ \boldsymbol{e}_{1}+2\boldsymbol{e}_{3}\right\} ,$
$\text{{\bf null}}(A_{14})=\text{{\bf span}}\left\{ \boldsymbol{e}_{3}\right\} ,$
${\color{black}\text{{\bf null}}(A_{23})=\text{{\bf span}}\left\{ \boldsymbol{e}_{2}\right\} }\mathpunct{\color{black},}$
$\text{{\bf null}}(A_{24})=\text{{\bf span}}\left\{ \boldsymbol{e}_{1}+\boldsymbol{e}_{2}+\boldsymbol{e}_{3}\right\} ,$
and $\text{{\bf null}}(A_{34})=\text{{\bf span}}\left\{ \boldsymbol{e}_{1}\right\} .$

Since the steady state difference $\boldsymbol{x}^{*}_{i}-\boldsymbol{x}^{*}_{j}$
lies in the null space of $A_{ij}$ for every edge $(i,j)\in\mathcal{E}$,
it is orthogonal to the row space of each edge weight matrix. Consequently,
$\boldsymbol{x}^{*}_{i}-\boldsymbol{x}^{*}_{j}$ is orthogonal to
the intersection of all edge row spaces. If this intersection contains
a nonzero subspace $\mathbb{V}$, i.e.,
\begin{equation}
\mathbb{V}\subseteq\bigcap_{(i,j)\in\mathcal{E}}\mathbf{row}(A_{ij}),
\end{equation}
then the projection of $\boldsymbol{x}^{*}_{i}-\boldsymbol{x}^{*}_{j}$
onto $\mathbb{V}$ vanishes for every edge. By connectivity of the
underlying graph, this implies that the projections of all agent states
onto $\mathbb{V}$ coincide, and hence subspace consensus on $\mathbb{V}$
is achieved. However, in Example \ref{iterative algorithm-motivation example},
a direct inspection of the row spaces of all edge weight matrices
yields
\begin{equation}
\bigcap_{(i,j)\in\mathcal{E}}\mathbf{row}(A_{ij})=\{\mathbf{0}\},
\end{equation}
which indicates that there is no common nonzero direction in the row
spaces of all edges. From this local information alone, it appears
that no nontrivial subspace can be enforced simultaneously by every
edge. However, \textcolor{black}{as shown in}\textcolor{red}{{} }Figure
\ref{4 vertex trajectory}, subspace consensus is still achievable,
\begin{figure}[tbh]
\begin{centering}
\includegraphics[width=8cm]{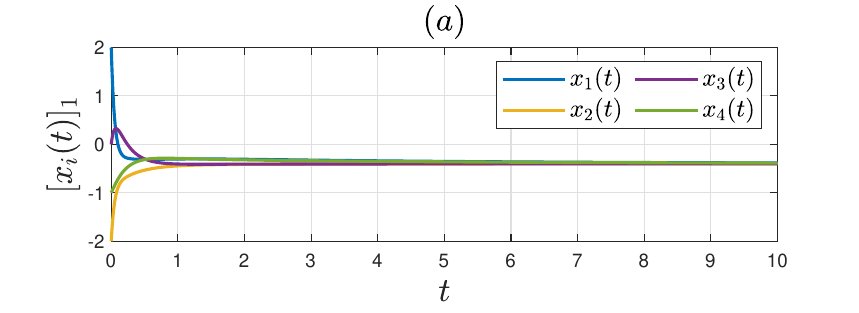}
\par\end{centering}
\begin{centering}
\includegraphics[width=8cm]{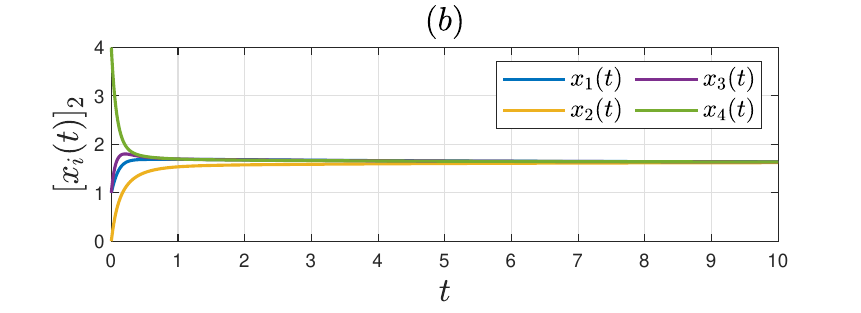}
\par\end{centering}
\begin{centering}
\includegraphics[width=8cm]{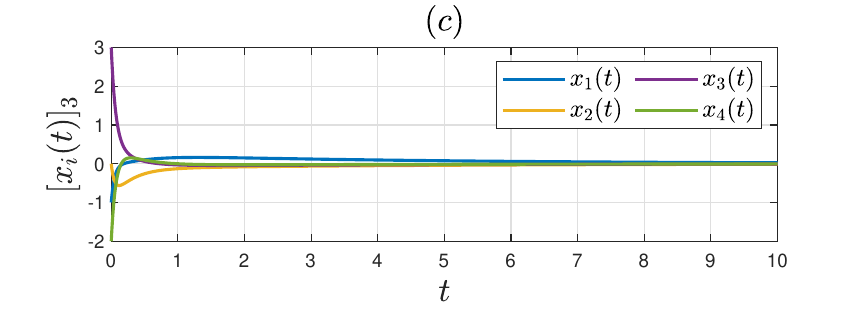}
\par\end{centering}
\centering{}\caption{\label{4 vertex trajectory}The dimension-wise trajectories of the
matrix-weighted network \ref{eq:overall} with network $\mathcal{G}$
shown in Figure \ref{The 4-node network}\textcolor{blue}{.}}
\end{figure}
 and moreover, the achieved consensus subspace is the entire space
$\mathbb{R}^{3}$. This observation suggests that the intersection
of the edge row spaces alone is insufficient to determine whether
global subspace consensus is achievable; the coupling constraints
inherent in the network may play a crucial role in the propagation
of consensus. To further investigate this issue, we will introduce
in the following a path-based iterative construction that explicitly
extracts the constraints encoded in the global topological structure.
\end{example}
\begin{rem}
It should be noted that the matrix $A_{ij}$ with a prescribed null
space $\mathbf{null}\left(A_{ij}\right)$ is not unique. In fact,
any invertible linear transformation applied to the row space of $A_{ij}$
leaves its null space unchanged. However, if we additionally require
$A_{ij}$ to be an orthogonal projection matrix, then it is uniquely
determined. Specifically, let $\ensuremath{N_{ij}\in\mathbb{R}^{d\times k}}$
be a matrix whose columns form a basis of the null space of the matrix
$A_{ij}$, where $k=\mathbf{\dim}\left(\mathbf{null}\left(A_{ij}\right)\right)$.
Then the unique orthogonal projection matrix with the prescribed null
space $\mathbf{null}\left(A_{ij}\right)$ is directly constructed
as:
\begin{equation}
A_{ij}=I_{d}-N_{ij}(N^{\top}_{ij}N_{ij})^{-1}N^{\top}_{ij}.
\end{equation}
\end{rem}
\textcolor{black}{Let $\ensuremath{\mathcal{Q}_{ij}=\{q^{ij}_{1},q^{ij}_{2},\ldots,q^{ij}_{m_{ij}}\}}$
denote the set of all simple paths from $i$ to $j$, where $m_{ij}=|\mathcal{Q}_{ij}|$.
For any path $q^{ij}_{s}\in\mathcal{Q}_{ij}$ with $s\in\underline{m_{ij}}$,
we write
\begin{equation}
q^{ij}_{s}=\left(v^{(s,ij)}_{0},v^{(s,ij)}_{1},\ldots,v^{(s,ij)}_{\ell^{ij}_{s}}\right),
\end{equation}
where $\ensuremath{v^{(s,ij)}_{0}=i},\ensuremath{v^{(s,ij)}_{\ell^{ij}_{s}}=j},$$\ell^{ij}_{s}\triangleq\ell(q^{ij}_{s})$
denotes the length of path $q^{ij}_{s}$ (the number of edges), $v^{(s,ij)}_{r}$
denotes the $\ensuremath{r}$-th node on path $q^{ij}_{s}$ with $0\le r\le\ell^{ij}_{s}$. }

\textcolor{black}{Let $T^{(0)}_{ij}=\mathbf{null}\left(A_{ij}\right)$
and denote by $T^{(l)}_{ij}\subseteq\mathbb{R}^{d}$ as the candidate
subspace associated with edge $(i,j)$ for each iteration index $l\in\mathbb{N}$.
With this notation, the candidate subspace of the path $q^{ij}_{s}$
is then defined as
\begin{equation}
\mathcal{W}^{(l)}(q^{ij}_{s})=\sum^{\ell^{ij}_{s}-1}_{r=0}T^{(l)}_{v^{(s,ij)}_{r}v^{(s,ij)}_{r+1}},
\end{equation}
for any $s\in\underline{m_{ij}}$ and $l\in\mathbb{N}$. Then the
update rule of the candidate subspace of edge $\left(i,j\right)$
can be expressed as
\begin{equation}
T^{(l+1)}_{ij}=\bigcap_{q^{ij}_{s}\in\mathcal{Q}_{ij}}\mathcal{W}^{(l)}(q^{ij}_{s}).
\end{equation}
}

The iterative construction above defines an update rule for the candidate
subspaces associated with each edge in the network. We now turn to
the question of whether this construction terminates in finitely many
steps, that is, whether there exists a finite integer $L$ such that
the candidate subspace remains invariant after $L$ updates for each
edge.
\begin{thm}
\textcolor{blue}{\label{T-convergence} }\textcolor{black}{Let $\mathcal{G}=(\mathcal{V},\mathcal{E},A)$
be a connected matrix-weighted network. Then there exists a finite
integer $L\ge0$ such that
\begin{equation}
T^{(l+1)}_{ij}=T^{(l)}_{ij},
\end{equation}
for all $l\ge L$ and for all edges $(i,j)\in\mathcal{E}$. }
\end{thm}
\begin{proof}
\textcolor{black}{Since $\mathcal{G}$ has finitely many vertices,
$\mathcal{Q}_{ij}$ is a finite set. Indeed, the length-one path $(i,j)$
is always contained in $\mathcal{Q}_{ij}$, since every edge is itself
a simple path from $i$ to $j$. Therefore, 
\begin{equation}
T^{(l+1)}_{ij}\subseteq T^{(l)}_{ij},
\end{equation}
for any $l\in\mathbb{N}$. Define the potential function
\begin{equation}
\Phi_{l}\triangleq\sum_{(i,j)\in\mathcal{E}}\mathbf{\dim}\left(T^{(l)}_{ij}\right),
\end{equation}
where $\mathbf{\dim}(\cdot)$ denotes the dimension of a linear subspace.
Since each $T^{(l)}_{ij}$ is a subspace of $\mathbb{R}^{d}$ and
the number of edges is finite, $\Phi_{l}$ is a nonnegative integer
for any $l\in\mathbb{N}$. For each edge $(i,j)\in\mathcal{E}$,
\begin{equation}
\mathbf{\mathbf{\dim}}\left(T^{(l+1)}_{ij}\right)\le\dim\left(T^{(l)}_{ij}\right).
\end{equation}
Summing over all edges yields
\begin{equation}
\Phi_{l+1}\le\Phi_{l}.
\end{equation}
Thus, $\{\Phi_{l}\}^{\infty}_{l=0}$ is a non-increasing sequence
of nonnegative integers. We claim that if there does not exist a finite
integer $L\ge0$ such that
\begin{equation}
T^{(l+1)}_{ij}=T^{(l)}_{ij},
\end{equation}
for all $l\ge L$ and for all edges $(i,j)\in\mathcal{E}$, then
\begin{equation}
\Phi_{l+1}<\Phi_{l}.
\end{equation}
By the claim, there exists at least one edge $(i_{0},j_{0})\in\mathcal{E}$
such that
\begin{equation}
T^{(l+1)}_{i_{0}j_{0}}\subsetneq T^{(l)}_{i_{0}j_{0}}.
\end{equation}
For such an edge, strict subspace inclusion implies strict dimension
decrease
\begin{equation}
\dim\left(T^{(l+1)}_{i_{0}j_{0}}\right)<\dim\left(T^{(l)}_{i_{0}j_{0}}\right).
\end{equation}
For all other edges, the dimensions do not increase. Hence, at least
one term in the sum defining $\Phi_{l+1}$ is strictly smaller than
its counterpart in $\Phi_{l}$, and no term is larger. Therefore,
$\Phi_{l+1}<\Phi_{l}$.}

\textcolor{black}{Denote by the initial value of the potential function
is
\begin{equation}
\Phi_{0}=\sum_{(i,j)\in\mathcal{E}}\dim\left(\operatorname{null}(A_{ij})\right),
\end{equation}
which is a finite nonnegative integer, since $\mathcal{G}$ has finitely
many edges and each $A_{ij}\in\mathbb{R}^{d\times d}$ has nullity
at most $d$. $\{\Phi_{l}\}$ is a non-increasing sequence of nonnegative
integers that strictly decreases under the above claim. Since $\Phi_{0}$
is finite, the sequence can strictly decrease at most $\Phi_{0}$
times. Therefore, there exists a finite integer $L\le\Phi_{0},$ such
that
\begin{equation}
T^{(L+1)}_{ij}=T^{(L)}_{ij},\quad\forall(i,j)\in\mathcal{E}.
\end{equation}
}
\end{proof}
\begin{defn}
\label{T^* definition}By Theorem \ref{T-convergence}, the limiting
candidate subspace $T^{*}_{ij}$ associated with edge $(i,j)$ is
defined as
\begin{equation}
T^{*}_{ij}=\lim_{l\rightarrow\infty}T^{(l)}_{ij}.
\end{equation}
\textcolor{blue}{}
\end{defn}
\begin{example}
\textcolor{black}{We now apply the iterative construction to the matrix-weighted
network in Example \ref{iterative algorithm-motivation example}.
We initialize $T^{(0)}_{ij}=\mathbf{null}\left(A_{ij}\right)$ for
each edge $(i,j)\in\mathcal{E}$ and demonstrate the iterative process
using edge $(1,2)$ as an example. There are five simple paths from
agent $1$ and $2$, $q^{12}_{1}=(1,2),$ $q^{12}_{2}=(1,3,2),$ $q^{12}_{3}=(1,4,2),$
$q^{12}_{4}=(1,3,4,2),$ and $q^{12}_{5}=(1,4,3,2).$}

\textcolor{black}{For $l=0$, we compute the path subspaces:
\begin{equation}
\mathcal{W}^{(0)}(q^{12}_{1})=\text{{\bf span}}\left\{ \boldsymbol{e}_{1}+\boldsymbol{e}_{2}+2\boldsymbol{e}_{3}\right\} ,
\end{equation}
\begin{equation}
\mathcal{W}^{(0)}(q^{12}_{2})=T^{0}_{13}+T^{0}_{23}=\text{{\bf span}}\left\{ \boldsymbol{e}_{1}+2\boldsymbol{e}_{3},\boldsymbol{e}_{2}\right\} ,
\end{equation}
\begin{equation}
\mathcal{W}^{(0)}(q^{12}_{3})=T^{0}_{14}+T^{0}_{24}=\text{{\bf span}}\left\{ \boldsymbol{e}_{1}+\boldsymbol{e}_{2}+\boldsymbol{e}_{3},\boldsymbol{e}_{3}\right\} ,
\end{equation}
\begin{equation}
\mathcal{W}^{(0)}(q^{12}_{4})=T^{0}_{13}+T^{0}_{34}+T^{0}_{24}=\mathbb{R}^{3},
\end{equation}
and
\begin{equation}
\mathcal{W}^{(0)}(q^{12}_{5})=T^{0}_{14}+T^{0}_{34}+T^{0}_{23}=\mathbb{R}^{3}.
\end{equation}
Hence, 
\begin{equation}
T^{1}_{12}=\bigcap^{5}_{s=1}\mathcal{W}^{(0)}(q^{12}_{s})=\text{{\bf span}}\left\{ \boldsymbol{e}_{1}+\boldsymbol{e}_{2}+2\boldsymbol{e}_{3}\right\} .
\end{equation}
 Applying the same procedure to the remaining edges, we obtain 
\begin{equation}
T^{1}_{13}=\text{{\bf span}}\left\{ \boldsymbol{e}_{1}+2\boldsymbol{e}_{3}\right\} ,T^{1}_{14}=\text{{\bf span}}\left\{ \boldsymbol{e}_{3}\right\} ,
\end{equation}
and
\begin{equation}
T^{1}_{23}=T^{1}_{24}=T^{1}_{34}=\{\boldsymbol{0}\}.
\end{equation}
Since $T^{1}_{23}\neq T^{0}_{23}$, a change occurs, and we proceed
to the next iteration. For $l=1$, applying the same update rule to
every edge yields,
\begin{equation}
T^{2}_{12}=T^{2}_{13}=T^{2}_{14}=T^{2}_{23}=T^{2}_{24}=T^{2}_{34}=\{\boldsymbol{0}\}.
\end{equation}
Since $T^{2}_{12}\neq T^{1}_{12}$, a change occurs, and we proceed
to the next iteration. For $l=3$, applying the same update rule to
every edge yields, 
\begin{equation}
T^{3}_{12}=T^{3}_{13}=T^{3}_{14}=T^{3}_{23}=T^{3}_{24}=T^{3}_{34}=\{\boldsymbol{0}\}.
\end{equation}
No further changes occur, so the iterative process terminates. Therefore,
we set $T^{*}_{ij}=T^{2}_{ij}$ for each edge.}
\end{example}
One can see that Theorem \ref{T-convergence} guarantees that the
iterative construction terminates in finitely many steps, so that
$T^{*}_{ij}$ is well-defined. The following Lemma \ref{intersection}
further reveals the intrinsic relation between $T^{*}_{ij}$ and the
steady state behavior of the system, it shows that the steady state
difference between any two agents must lie in the limiting candidate
subspace associated with the corresponding edge.
\begin{thm}
\textcolor{violet}{\label{intersection}}\textcolor{black}{{} Let $\boldsymbol{x}^{*}=\lim_{t\rightarrow\infty}\boldsymbol{x}(t)$
be the steady state of the matrix-weighted network \eqref{eq:overall}
where $\boldsymbol{x}^{*}=[\boldsymbol{x}^{*}_{1},\boldsymbol{x}^{*}_{2},\cdots,\boldsymbol{x}^{*}_{n}]^{\top}$.
Then 
\begin{equation}
\boldsymbol{x}^{*}_{i}-\boldsymbol{x}^{*}_{j}\in T^{*}_{ij},
\end{equation}
for any edge $(i,j)\in\mathcal{E}$.}
\end{thm}
\begin{proof}
\textcolor{black}{Since $\boldsymbol{x}^{*}=\lim_{t\rightarrow\infty}\boldsymbol{x}(t)$,
thus one has $\boldsymbol{x}^{*}\in\text{{\bf null}}(L)$. According
to Lemma \ref{lem:null-space}, for any edge $(i,j)\in\mathcal{E}$,
one has $\boldsymbol{x}^{*}_{i}-\boldsymbol{x}^{*}_{j}\in\text{{\bf null}}(A_{ij})$.
For any simple path $q^{ij}_{s}=\left(v^{(s,ij)}_{0},v^{(s,ij)}_{1},\ldots,v^{(s,ij)}_{\ell^{ij}_{s}}\right)\in\mathcal{Q}_{ij}$,
where $\ensuremath{v^{(s,ij)}_{0}=i},\ensuremath{v^{(s,ij)}_{\ell^{ij}_{s}}=j},$$\ell^{ij}_{s}\triangleq\ell(q^{ij}_{s})$
denotes the length of path $q^{ij}_{s}$, $v^{(s,ij)}_{r}$ denotes
the $\ensuremath{r}$-th node on path $q^{ij}_{s}$ with $0\le r\le\ell^{ij}_{s}$,
$s\in\underline{m_{ij}}$,
\begin{eqnarray}
\boldsymbol{x}^{*}_{i}-\boldsymbol{x}^{*}_{j} & = & \stackrel{\ell^{ij}_{s}-1}{{\displaystyle \mathop{\sum}_{k=0}}}\left(\boldsymbol{x}^{*}_{v^{(s,ij)}_{k+1}}-\boldsymbol{x}^{*}_{v^{(s,ij)}_{k}}\right)\\
 & \in & {\displaystyle \mathop{\sum}^{\ell^{ij}_{s}-1}_{k=0}}\text{{\bf null}}(A_{v^{(s,ij)}_{k}v^{(s,ij)}_{k+1}})=\mathcal{W}^{(0)}(q^{ij}_{s}).
\end{eqnarray}
Therefore, $\boldsymbol{x}^{*}_{i}-\boldsymbol{x}^{*}_{j}\in\mathcal{W}^{(0)}(q^{ij}_{s})$.
Then,
\begin{equation}
\boldsymbol{x}^{*}_{i}-\boldsymbol{x}^{*}_{j}\in\bigcap_{q^{ij}_{s}\in\mathcal{Q}_{ij}}\mathcal{W}^{(0)}(q^{ij}_{s}).
\end{equation}
Since 
\begin{equation}
T^{(1)}_{ij}=\bigcap_{q^{ij}_{s}\in\mathcal{Q}_{ij}}\mathcal{W}^{(0)}(q^{ij}_{s}),
\end{equation}
then one has $\boldsymbol{x}^{*}_{i}-\boldsymbol{x}^{*}_{j}\in T^{(1)}_{ij}$
for any edge $(i,j)\in\mathcal{E}$. }
\end{proof}
\textcolor{black}{We now prove by induction that $\boldsymbol{x}^{*}_{i}-\boldsymbol{x}^{*}_{j}\in T^{(l)}_{ij}$
for all $l\in\mathbb{N}$. Suppose that for some $l\geq1$,}

\textcolor{black}{
\begin{equation}
\boldsymbol{x}^{*}_{i}-\boldsymbol{x}^{*}_{j}\in T^{(l)}_{ij}.
\end{equation}
For any simple path $q^{ij}_{s}=\left(v^{(s,ij)}_{0},v^{(s,ij)}_{1},\ldots,v^{(s,ij)}_{\ell^{ij}_{s}}\right)\in\mathcal{Q}_{ij}$,
by the induction hypothesis applied to each edge $(v^{(s,ij)}_{k},v^{(s,ij)}_{k+1})$,}

\textcolor{black}{
\begin{equation}
\boldsymbol{x}^{*}_{v^{(s,ij)}_{k+1}}-\boldsymbol{x}^{*}_{v^{(s,ij)}_{k}}\in T^{(l)}_{v^{(s,ij)}_{k}v^{(s,ij)}_{k+1}}.
\end{equation}
Thus,
\begin{eqnarray}
\boldsymbol{x}^{*}_{i}-\boldsymbol{x}^{*}_{j} & = & \stackrel{\ell^{ij}_{s}-1}{{\displaystyle \mathop{\sum}_{k=0}}}\left(\boldsymbol{x}^{*}_{v^{(s,ij)}_{k+1}}-\boldsymbol{x}^{*}_{v^{(s,ij)}_{k}}\right)\\
 & \in & {\displaystyle \mathop{\sum}^{\ell^{ij}_{s}-1}_{k=0}}T^{(l)}_{v^{(s,ij)}_{k}v^{(s,ij)}_{k+1}}=\mathcal{W}^{(l)}(q^{ij}_{s}).
\end{eqnarray}
Again, since $q^{ij}_{s}$ is arbitrary,}

\textcolor{black}{
\begin{equation}
\boldsymbol{x}^{*}_{i}-\boldsymbol{x}^{*}_{j}\in\bigcap_{q^{ij}_{s}\in\mathcal{Q}_{ij}}\mathcal{W}^{(l)}(q^{ij}_{s})=T^{(l+1)}_{ij}.
\end{equation}
By induction, $\boldsymbol{x}^{*}_{i}-\boldsymbol{x}^{*}_{j}\in T^{(l)}_{ij}$
for every $l\in\mathbb{N}$. Since the sequence $\{T^{(l)}_{ij}\}^{\infty}_{l=0}$
is monotonically decreasing (by construction, $T^{(l+1)}_{ij}\subseteq T^{(l)}_{ij}$)
and converges to $T^{*}_{ij}$, it follows that}

\textcolor{black}{
\begin{equation}
\boldsymbol{x}^{*}_{i}-\boldsymbol{x}^{*}_{j}\in T^{*}_{ij},
\end{equation}
for any edge $(i,j)\in\mathcal{E}$.}
\begin{defn}
\textcolor{black}{Let $\mathcal{G}=(\mathcal{V},\mathcal{E},A)$ be
a connected matrix-weighted network. The limiting constraint network
of $\mathcal{G}$, denoted by $\mathcal{G}^{*}=(\mathcal{V},\mathcal{E},A^{*})$,
is defined as the matrix-weighted network with the same node and edge
sets as $\mathcal{G}$, whose weight matrix $A^{*}=(A^{*}_{ij})\in\mathbb{R}^{dn\times dn}$
satisfy 
\begin{equation}
\mathbf{null}\left(A^{*}_{ij}\right)=T^{*}_{ij}.
\end{equation}
}
\end{defn}
\begin{defn}
For any spanning tree $\mathcal{T}$ in $\mathcal{G}$, the achievable
subspace associated with $\mathcal{T}$, denoted by $\mathbb{V}_{\mathsf{ach}}(\mathcal{T})$,
is defined as the intersection of the row spaces of all its edges:

\begin{align}
\mathbb{V}_{\mathsf{ach}}\left(\mathcal{T}\right) & \triangleq\bigcap_{(i,j)\in\mathcal{E}(\mathcal{T})}\mathbf{row}\left(A_{ij}\right)\\
 & =\left({\displaystyle \sum_{(i,j)\in\mathcal{E}(\mathcal{T})}}\mathbf{null}\left(A_{ij}\right)\right)^{\bot}.
\end{align}
\end{defn}
\begin{rem}
\textcolor{black}{The limiting constraint network $\mathcal{G}^{*}$
encodes the limiting candidate subspace $T^{*}_{ij}$ on each edge
as the null space of the corresponding edge weight matrix. In this
way, the constraints that originally depended on global path information
are compressed into local null-space constraints on each edge. The
}matrix-weighted network\textcolor{black}{{} \eqref{eq:overall} achieves
subspace consensus on $\mathbb{V}\subseteq\mathbb{R}^{d}$ if $\mathcal{G}$
has a $\mathbb{V}$-spanning tree. }Consequently, one can see that
the system can achieve subspace consensus on any subspace $\mathbb{V}$
satisfying $\mathbb{V}\subseteq\mathbb{V}_{\mathsf{ach}}(\mathcal{T})$.\textcolor{black}{{}
However, A natural question then arises: do different spanning trees
of $\mathcal{G}$ lead to different conclusions regarding the achievable
consensus subspace? }
\end{rem}
\begin{example}
\label{exa: 3 vertex graph-1}\textcolor{violet}{} Consider the matrix-weighted
network \eqref{eq:overall} with the network $\mathcal{G}$ shown
in Figure \textcolor{violet}{\ref{fig:The-network-structure of exm for def}},
\textcolor{black}{where the state of each agent is three-dimensional.
}Let $\mathbb{V}=\mathbf{span}\{\boldsymbol{e}_{1}-2\boldsymbol{e}_{3},\boldsymbol{e}_{2}-2\boldsymbol{e}_{3}\}$,
the row spaces of the coupling matrices are configured as 
\begin{equation}
\mathbf{row}(A_{25})=\mathbf{row}(A_{23})=\mathbb{V},
\end{equation}
 and\textcolor{violet}{{} }
\begin{align}
\text{{\bf row}}(A_{14}) & =\text{{\bf row}}(A_{45})=\text{{\bf row}}(A_{16})\nonumber \\
 & =\text{{\bf row}}(A_{12})=\text{{\bf row}}(A_{13})=\mathbb{R}^{3}.
\end{align}
Under this configuration, consider the spanning tree
\begin{equation}
\mathcal{T}_{1}=\{(5,2),(2,3),(1,4),(4,5),(1,6)\},
\end{equation}
the achievable subspace of $\mathcal{T}_{1}$ is $\mathbb{V}$. In
contrast, consider the spanning tree
\begin{equation}
\mathcal{T}_{2}=\{(1,2),(1,3),(1,4),(4,5),(1,6)\},
\end{equation}
which has all its edges with row space $\mathbb{R}^{3}$, hence the
achievable subspace of $\mathcal{T}_{2}$ is $\mathbb{R}^{3}$.
\end{example}
\begin{rem}
As illustrated by the Example \ref{exa: 3 vertex graph-1}, two different
spanning trees yield different achievable subspaces. Since the system
can achieve subspace consensus on any subspace $\mathbb{V}$ satisfying
$\mathbb{V}\subseteq\mathbb{V}_{\mathsf{ach}}(\mathcal{T})$, the
achievable subspace directly bounds the set of consensus subspaces
that can be certified from the chosen spanning tree. Given that different
spanning trees generally yield different achievable subspaces, the
resulting conclusions regarding subspace consensus may also vary.
Further, does the choice of spanning tree of $\mathcal{G}^{*}$affect
the conclusion regarding the achievable consensus subspace? The following
Theorem \ref{spanning tree  subspace} shows that the answer is negative:
the orthogonal complement of the sum of $T^{*}_{ij}$ over any spanning
tree of $\mathcal{G}^{*}$ is independent of the choice of the tree.
This invariance property reveals that the limiting constraint network
possesses an intrinsic tree-independent structure, which lays the
foundation for the subsequent spanning-tree-based characterization
of subspace consensus.
\end{rem}
\begin{thm}
\textcolor{violet}{\label{spanning tree  subspace}}\textcolor{black}{Let
$\mathcal{T}_{1}$ and $\mathcal{T}_{2}$ be any two spanning trees
of $\mathcal{G}^{*}=(\mathcal{V},\mathcal{E},A^{*})$. Then
\begin{equation}
\mathbb{V}_{\mathsf{ach}}\left(\mathcal{T}_{1}\right)=\mathbb{V}_{\mathsf{ach}}\left(\mathcal{T}_{2}\right).
\end{equation}
}
\end{thm}
\begin{proof}
\textcolor{black}{We prove the theorem by contradiction. Suppose that
\begin{equation}
\mathbb{V}_{\mathsf{ach}}\left(\mathcal{T}_{1}\right)\neq\mathbb{V}_{\mathsf{ach}}\left(\mathcal{T}_{2}\right),
\end{equation}
then there exists at least one nonzero vector $\boldsymbol{\xi}\in\mathbb{R}^{d}$
such that
\[
\boldsymbol{\xi}\in\mathbb{V}_{\mathsf{ach}}\left(\mathcal{T}_{1}\right)\quad\text{but}\quad\boldsymbol{\xi}\notin\mathbb{V}_{\mathsf{ach}}\left(\mathcal{T}_{2}\right).
\]
Since $\boldsymbol{\xi}\notin\mathbb{V}_{\mathsf{ach}}\left(\mathcal{T}_{2}\right)$,
by the definition of $\mathbb{V}_{\mathsf{ach}}\left(\mathcal{T}_{2}\right)$,
there exists at least one edge $(i_{0},j_{0})\in\mathcal{T}_{2}$
such that
\[
\boldsymbol{\xi}\notin\left(T^{*}_{i_{0}j_{0}}\right)^{\bot}.
\]
Equivalently, $\boldsymbol{\xi}$ is not orthogonal to $T^{*}_{i_{0}j_{0}}$.
Therefore, the projection of $\boldsymbol{\xi}$ onto $T^{*}_{i_{0}j_{0}}$
is nonzero, i.e., there exists a nonzero vector $\boldsymbol{\xi}_{2}\in T^{*}_{i_{0}j_{0}}$
such that $\boldsymbol{\ensuremath{\xi}}$ can be decomposed as
\begin{equation}
\boldsymbol{\xi}=\boldsymbol{\xi}_{1}+\boldsymbol{\xi}_{2},
\end{equation}
where
\begin{equation}
\boldsymbol{\xi}_{1}\in\left(T^{*}_{i_{0}j_{0}}\right)^{\bot}.
\end{equation}
Now consider the spanning tree $\mathcal{T}_{1}$. Since $\mathcal{T}_{1}$
is a spanning tree, there exists an unique simple path in $\mathcal{T}_{1}$
connecting $i_{0}$ and $j_{0}$. Denote this path by
\begin{equation}
P_{1}=(v_{0},v_{1},\ldots,v_{r}),
\end{equation}
where $\ensuremath{v_{0}=i_{0}},\ensuremath{v_{r}=j_{0}}$, and its
edges are
\begin{equation}
e_{l}=(v_{l-1},v_{l}),\quad l=1,2,\ldots,r.
\end{equation}
Since $\boldsymbol{\xi}\in\mathbb{V}_{\mathsf{ach}}\left(\mathcal{T}_{1}\right)$,
by the definition of $\mathbb{V}_{\mathsf{ach}}\left(\mathcal{T}_{1}\right)$,
$\boldsymbol{\xi}$ is orthogonal to $T^{*}_{e_{l}}$ for every edge
$e_{l}\in\mathcal{T}_{1}$. In particular, for each edge $e_{l}$
on the path $P_{1}$,
\begin{equation}
\boldsymbol{\xi}\perp T^{*}_{e_{l}},\,\forall l=1,\ldots,r.
\end{equation}
Next, we decompose $\boldsymbol{\xi}_{2}$ along the path $P_{1}$.
By the definition \ref{T^* definition}, we have
\begin{equation}
T^{*}_{i_{0}j_{0}}\subseteq\sum^{r}_{l=1}T^{*}_{e_{l}}.
\end{equation}
Therefore, since $\boldsymbol{\xi}_{2}\in T^{*}_{i_{0}j_{0}}$, there
exist vectors
\begin{equation}
\boldsymbol{\xi}_{2l}\in T^{*}_{e_{l}},\,l=\underline{r},
\end{equation}
such that
\begin{equation}
\boldsymbol{\xi}_{2}=\sum^{r}_{l=1}\boldsymbol{\xi}_{2l}.
\end{equation}
For each $l=1,\ldots,r$, since $\boldsymbol{\xi}\perp T^{*}_{e_{l}}$,
we have
\begin{equation}
\boldsymbol{\xi}\perp\boldsymbol{\xi}_{2l}.
\end{equation}
It follows that
\begin{equation}
\boldsymbol{\xi}^{\top}\boldsymbol{\xi}_{2}=\boldsymbol{\xi}^{\top}\left(\sum^{r}_{l=1}\boldsymbol{\xi}_{2l}\right)=\sum^{r}_{l=1}\boldsymbol{\xi}^{\top}\boldsymbol{\xi}_{2l}=0.
\end{equation}
Hence,
\begin{equation}
\boldsymbol{\xi}\perp\boldsymbol{\xi}_{2}.
\end{equation}
On the other hand, since $\boldsymbol{\xi}=\boldsymbol{\xi}_{1}+\boldsymbol{\xi}_{2}$
and $\boldsymbol{\xi}_{1}\perp\boldsymbol{\xi}_{2}$, we have
\begin{equation}
0=\boldsymbol{\xi}^{\top}\boldsymbol{\xi}_{2}=\left(\boldsymbol{\xi}_{1}+\boldsymbol{\xi}_{2}\right)^{\top}\boldsymbol{\xi}_{2}=\boldsymbol{\xi}^{\top}_{2}\boldsymbol{\xi}_{2},
\end{equation}
which implies $\boldsymbol{\xi}_{2}=\boldsymbol{0}$. Therefore $\boldsymbol{\xi}=\boldsymbol{\xi}_{1}\in\left(T^{*}_{i_{0}j_{0}}\right)^{\bot}$,
which is contradict with the assumption $\mathbb{V}_{\mathsf{ach}}\left(\mathcal{T}_{1}\right)\neq\mathbb{V}_{\mathsf{ach}}\left(\mathcal{T}_{2}\right)$,
and we have $\mathbb{V}_{\mathsf{ach}}\left(\mathcal{T}_{1}\right)=\mathbb{V}_{\mathsf{ach}}\left(\mathcal{T}_{2}\right)$.}
\end{proof}
\textcolor{black}{}
\begin{cor}
\textcolor{black}{Let $\mathcal{T}_{1}$ and $\mathcal{T}_{2}$ be
any two spanning trees of $\mathcal{G}^{*}=(\mathcal{V},\mathcal{E},A^{*})$.
Then $\mathcal{T}_{1}$ is a positive spanning tree if and only if
$\mathcal{T}_{2}$ is a positive spanning tree.}
\end{cor}
\begin{proof}
\textcolor{black}{Denote by 
\begin{equation}
\mathbb{V}_{\mathsf{ach}}\left(\mathcal{T}_{1}\right)=\left({\displaystyle \sum_{(i,j)\in\mathcal{E}\left(\mathcal{T}_{1}\right)}}T^{*}_{ij}\right)^{\bot},
\end{equation}
and
\begin{equation}
\mathbb{V}_{\mathsf{ach}}\left(\mathcal{T}_{2}\right)=\left({\displaystyle \sum_{(i,j)\in\mathcal{E}\left(\mathcal{T}_{2}\right)}}T^{*}_{ij}\right)^{\bot}.
\end{equation}
since $T^{*}_{ij}=\{\boldsymbol{0}\}$ for any edge $(i,j)\in\mathcal{E}\left(\mathcal{T}_{1}\right)$,
therefore $\mathbb{V}_{\mathsf{ach}}\left(\mathcal{T}_{1}\right)=\mathbb{R}^{d}$,
and according to the Theorem \ref{spanning tree  subspace}, we have
$\mathbb{V}_{\mathsf{ach}}\left(\mathcal{T}_{2}\right)=\mathbb{R}^{d}$,
i.e., $\left({\displaystyle \sum_{(i,j)\in\mathcal{E}\left(\mathcal{T}_{2}\right)}}T^{*}_{ij}\right)^{\bot}=\mathbb{R}^{d}$.
Thus we have $T^{*}_{ij}=\{\boldsymbol{0}\}$ for any edge $(i,j)\in\mathcal{E}\left(\mathcal{T}_{2}\right)$
and $\mathcal{T}_{2}$ is a positive spanning tree. The converse follows
by symmetry. }
\end{proof}
\textcolor{black}{}
\begin{lem}
\textcolor{black}{\label{row-space-for-edge}Let $\mathcal{T}$ be
a spanning tree of $\mathcal{G}^{*}=(\mathcal{V},\mathcal{E},A^{*})$.
Then $\mathbb{V}_{\mathsf{ach}}\left(\mathcal{T}\right)\subseteq\left(T^{*}_{ij}\right)^{\bot}$
for any edge $(i,j)\in\mathcal{E}$.}
\end{lem}
\begin{proof}
\textcolor{black}{Take any edge }\textcolor{violet}{$(i,j)\in\mathcal{E}$},
if $\left(i,j\right)\in\mathcal{T}$, the claim holds. If $(i,j)\notin\mathcal{T}$,
adding $(i,j)$ to $\mathcal{T}$ creates a unique cycle. Pick any
other edge $\left(i_{0},j_{0}\right)$ on this cycle different from
$(i,j)$, and set $\mathcal{T}'=\mathcal{T}\cup\{\left(i,j\right)\}\setminus\{\left(i_{0},j_{0}\right)\}$.
Then $\mathcal{T}'$ is still connected and has $|\mathcal{V}|-1$
edges, hence it is a spanning tree of $\mathcal{G}^{*}$ that contains
$\left(i,j\right)$. Since $\mathcal{T}'$ is also a spanning tree
of $\mathcal{G}^{*}$, by Theorem \ref{spanning tree  subspace},
\begin{equation}
\left({\displaystyle \sum_{(i,j)\in\mathcal{E}\left(\mathcal{T}'\right)}}T^{*}_{ij}\right)^{\bot}=\left({\displaystyle \sum_{(i,j)\in\mathcal{E}\left(\mathcal{T}\right)}}T^{*}_{ij}\right)^{\bot}.
\end{equation}
Moreover, since $(i,j)\in\mathcal{E}\left(\mathcal{T}'\right)$, we
have $\mathbb{V}_{\mathsf{ach}}\left(\mathcal{T}\right)\subseteq(T^{*}_{ij})^{\perp}$.
This completes the proof.

\textcolor{black}{Theorem \ref{spanning tree  subspace} and Lemma
\ref{row-space-for-edge} show that the achievable subspace associated
with any spanning tree $\mathcal{T}$ of $\mathcal{G}^{*}$ is the
same for every spanning tree, i.e., the achievable subspace is a common
subspace independent of the choice of $\mathcal{T}$. This common
subspace is defined as the orthogonal complement of the sum of $T^{*}_{ij}$
over any spanning tree $\mathcal{T}$, and it is contained in the
orthogonal complement of every edge subspace $T^{*}_{ij}$ in the
limiting constraint network. This property naturally implies that
any subspace contained in the achievable subspace is identical for
every spanning tree and also contained in the orthogonal complement
of every edge subspace $T^{*}_{ij}$. Consequently, a necessary and
sufficient condition for subspace consensus in terms of the existence
of a $\mathbb{V}$-spanning tree in $\mathcal{G}^{*}$ follows naturally,
as stated in the following theorem.} \textcolor{black}{}However,
for general matrix-weighted networks, the presence of multiple paths
between nodes introduces a fundamental difficulty: the basis vectors
of the null spaces associated with different edges are not necessarily
chosen from a common set of basis vectors, and hence the decomposition
of the same vector along different paths may be mutually inconsistent.
There is no natural coordinate correspondence between two such decompositions,
so they cannot be directly compared or combined. This \textquotedbl multi-path,
multi-decomposition\textquotedbl{} phenomenon makes it extremely difficult
to establish the necessity part of a necessary and sufficient condition
directly from the algebraic properties of the original network. To
partially alleviate this difficulty, in what follows we consider a
class of nontrivial problems in which the null space of every edge
weight matrix is assumed to be spanned by a subset of a common set
of linearly independent basis vectors.
\end{proof}
\begin{thm}
\textcolor{violet}{\label{sufficient and necessary condition}} Consider
the matrix-weighted network \eqref{eq:overall}\textcolor{black}{{}
with limiting constraint network $\mathcal{G}^{*}=(\mathcal{V},\mathcal{E},A^{*})$.
Then, the following statements are equivalent: }

\textcolor{black}{1) The }matrix-weighted network\textcolor{black}{{}
\eqref{eq:overall} achieves subspace consensus on ${\color{violet}{\color{black}\mathbb{V}\subseteq\mathbb{R}^{d}}}$.}

\textcolor{black}{2) The limiting constraint network $\mathcal{G}^{*}$contains
at least one $\mathbb{V}$-spanning tree.}

\textcolor{black}{3) Every spanning tree of $\mathcal{G}^{*}$is a
$\mathbb{V}$-spanning tree.}
\end{thm}
\begin{proof}
\textcolor{black}{$1\Longrightarrow2$: We prove this by contradiction.
Suppose }the matrix-weighted network\textcolor{black}{{} \eqref{eq:overall}
achieves subspace consensus on ${\color{violet}{\color{black}\mathbb{V}\subseteq\mathbb{R}^{d}}}$,
but $\mathcal{G}^{*}$ does not have a V-spanning tree. Then, there
exists at least one edge $(i_{0},j_{0})\in\mathcal{E}$ such that
$\mathbb{V}\nsubseteq\left(T^{*}_{i_{0}j_{0}}\right)^{\bot}$. That
is, there exists a non-zero vector $\pmb{\xi}\in\mathbb{V}$ that
can be decomposed as $\pmb{\xi}=\pmb{\xi}_{1}+\pmb{\xi}_{2}$, where
$\pmb{\xi}_{1}\in(T^{*}_{i_{0}j_{0}})^{\perp}$ and $\pmb{\xi}_{2}\in T^{*}_{i_{0}j_{0}}$
with $\pmb{\xi}_{2}\neq\mathbf{0}$. }
\end{proof}
\textcolor{black}{Let $\ensuremath{\mathcal{Q}_{ij}=\{q^{ij}_{1},q^{ij}_{2},\ldots,q^{ij}_{m_{ij}}\}}$
denote the set of all simple paths from $i$ to $j$, where $m_{ij}=|\mathcal{Q}_{ij}|$.
For any path $q^{ij}_{s}\in\mathcal{Q}_{ij}$ with $s\in\underline{m_{ij}}$,
we write
\begin{equation}
q^{ij}_{s}=\left(v^{(s,ij)}_{0},v^{(s,ij)}_{1},\ldots,v^{(s,ij)}_{\ell^{ij}_{s}}\right),
\end{equation}
where $\ensuremath{v^{(s,ij)}_{0}=i},\ensuremath{v^{(s,ij)}_{\ell^{ij}_{s}}=j},$$\ell^{ij}_{s}\triangleq\ell(q^{ij}_{s})$
denotes the length of path $q^{ij}_{s}$ (the number of edges), $v^{(s,ij)}_{r}$
denotes the $\ensuremath{r}$-th node on path $q^{ij}_{s}$ with $0\le r\le\ell^{ij}_{s}$.
For any path $q^{ij}_{s}$ connecting node $i$ and $j$ in the network
$\mathcal{G}^{*}$, the candidate subspace of $q^{ij}_{s}$ is defined
as
\begin{equation}
W^{*}\left(q^{ij}_{s}\right)=\sum^{\ell^{ij}_{s}-1}_{r=0}T^{*}_{v^{(s,ij)}_{r}v^{(s,ij)}_{r+1}}.
\end{equation}
We choose the initial state $\boldsymbol{x}(0)=\left[\boldsymbol{x}^{^{\top}}_{1}(0),\,,\boldsymbol{x}^{^{\top}}_{2}(0),\,\cdots,\,\boldsymbol{x}^{^{\top}}_{n}(0)\right]^{\top}$such
that $\boldsymbol{x}_{i_{0}}(0)=\boldsymbol{0}$. For any node $v\in\mathcal{V}$,
the candidate subspace $S_{v}$ of $v$ is defined as the intersection
of the null spaces of all paths from $i_{0}$ to $v$
\begin{equation}
S_{v}=\bigcap_{q^{i_{0}v}_{s}\in\mathcal{Q}_{i_{0}v}}W^{*}\left(q^{i_{0}v}_{s}\right).
\end{equation}
For any edge $(u,v)\in\mathcal{E}$, we shall first show that
\begin{equation}
S_{u}\triangle S_{v}\subseteq T^{*}_{uv}.
\end{equation}
Take any element $\boldsymbol{w}\in S_{u}\triangle S_{v}$, by the
definition of symmetric difference, $\boldsymbol{w}$ belongs to exactly
one of $S_{u}$ or $S_{v}$. Then we consider the following two cases:}

\textcolor{black}{$\textbf{Case 1:}$ ($\boldsymbol{w}\in S_{u}$
and $\boldsymbol{w}\notin S_{v}$). Since $\boldsymbol{w}\in S_{u}$,
by the Definition of $S_{u}$, for every path $q^{i_{0}u}_{s}$ from
$i_{0}$ to $u$, we have
\begin{equation}
\boldsymbol{w}\in W^{*}\left(q^{i_{0}u}_{s}\right).
\end{equation}
Since $\boldsymbol{w}\notin S_{v}$, there exists at least one path
$Q_{0}$ from $i_{0}$ to $v$ such that
\begin{equation}
\boldsymbol{w}\notin W^{*}\left(Q_{0}\right).
\end{equation}
Because $(u,v)$ is an edge and the graph is connected, denote by
the path $Q$ as
\begin{equation}
Q=Q_{0}\cup\{(u,v)\},
\end{equation}
one can see that $Q$ is a path from $i_{0}$ to $u$. By the Definition
of the space of $Q$, it satisfies
\begin{equation}
W^{*}\left(Q\right)=W^{*}\left(Q_{0}\right)+T^{*}_{uv}.
\end{equation}
Since $\boldsymbol{w}\notin W^{*}\left(Q_{0}\right)$ and $\boldsymbol{w}\in W^{*}\left(Q\right)$,
hence
\begin{equation}
\boldsymbol{w}\in T^{*}_{uv}.
\end{equation}
}

\textcolor{black}{$\textbf{Case 2:}$ ($\boldsymbol{w}\in S_{v}$
and $\boldsymbol{w}\notin S_{u}$). By the same argument with $u$
and $v$ interchanged, we obtain
\begin{equation}
\boldsymbol{w}\in T^{*}_{uv}.
\end{equation}
Thus, in both cases, $\boldsymbol{w}\in T^{*}_{uv}$. Since $\boldsymbol{w}$
was arbitrary, therefore we conclude
\begin{equation}
S_{u}-S_{v}\subseteq T^{*}_{uv}.
\end{equation}
}

Since the basis of $\mathbf{null}\left(A_{ij}\right)$ is a subset
of $\Phi=\left\{ \boldsymbol{\phi}_{1},\boldsymbol{\phi}_{2},\cdots,\boldsymbol{\phi}_{m}\right\} $
for any $\left(i,j\right)\in\mathcal{E}$, where the vectors $\boldsymbol{\phi}_{i},\;i\in\underline{m}$
are linearly independent. Therefore, both the basis of $T^{*}_{ij}$
and $S_{v}$ belong to $\Phi$ . Let the state of node $v$ be defined
as the sum of the basis vectors of its node space

\begin{equation}
\boldsymbol{x}_{v}=\sum_{\boldsymbol{\phi}_{i}\in\operatorname{Basis}(S_{v})}\boldsymbol{\phi}_{i},
\end{equation}
and

\begin{equation}
\boldsymbol{x}_{u}\triangle\boldsymbol{x}_{v}\subseteq T^{*}_{uv}.
\end{equation}

In the following, we shall prove that $\boldsymbol{x}_{j_{0}}\neq\boldsymbol{0}$.
Since $\pmb{\xi}_{2}\in T^{*}_{i_{0}j_{0}}$ with $\pmb{\xi}_{2}\neq\mathbf{0}$,
denote by $\pmb{\xi}_{2}=p_{1}\boldsymbol{\phi}_{i_{1}}+p_{2}\boldsymbol{\phi}_{i_{2}}+\cdots+p_{i_{k}}\boldsymbol{\phi}_{i_{k}}$,
where $\boldsymbol{\phi}_{i_{1}},\boldsymbol{\phi}_{i_{2}},\cdots,\boldsymbol{\phi}_{i_{k}}\in T^{*}_{i_{0}j_{0}}$,
and we can choose $\boldsymbol{x}_{j_{0}}=\boldsymbol{\phi}_{i_{1}}+\boldsymbol{\phi}_{i_{2}}+\cdots+\boldsymbol{\phi}_{i_{k}}$.
I\textcolor{black}{n summary, we have successfully constructed an
initial state $\pmb{x}(0)$ such that $L\pmb{x}(0)=\mathbf{0}$. Consequently,
the trajectory of }the matrix-weighted network \eqref{eq:overall}
re\textcolor{black}{mains frozen at $\pmb{x}(0)$, which implies that
the steady state is exactly $\pmb{x}^{*}=\pmb{x}(0)$, i.e.,
\begin{equation}
\mathbf{P}_{\mathbb{V}}\big(\pmb{x}^{*}_{i_{0}}-\pmb{x}^{*}_{j_{0}}\big)=\mathbf{P}_{\mathbb{V}}\big(\pmb{x}_{i_{0}}(0)-\pmb{x}_{j_{0}}(0)\big)=\mathbf{P}_{\mathbb{V}}(-\pmb{\xi}_{2})\neq\mathbf{0}.
\end{equation}
}

\textcolor{black}{This implies that the projections of the states
of agents $i_{0}$ and $j_{0}$ onto the subspace $\mathbb{V}$ do
not coincide. Therefore, }the matrix-weighted network \eqref{eq:overall}
fail\textcolor{black}{s to achieve subspace consensus on $\mathbb{V}$,
which contradicts the assumption.}

\textcolor{black}{$2\Longrightarrow1$: Denote by $\boldsymbol{\lim}_{t\rightarrow\infty}\boldsymbol{x}(t)=\boldsymbol{x}^{*}$
and $\boldsymbol{x}^{*}=\left[\boldsymbol{x}^{*^{\top}}_{1},\,,\boldsymbol{x}^{*^{\top}}_{2},\,\cdots,\,\boldsymbol{x}^{*^{\top}}_{n}\right]^{\top}$.
Since $\boldsymbol{x}^{*}_{i}-\boldsymbol{x}^{*}_{j}\in T^{*}_{ij}$
for any $(i,j)\in\mathcal{E}$ and $\mathcal{G}^{*}$ has a $\mathbb{V}-$spanning
tree $\mathcal{T}$, according to Lemma \ref{row-space-for-edge},
$\mathbb{V}\subseteq\mathbb{V}_{\mathsf{ach}}\left(\mathcal{T}\right)\subseteq\left(T^{*}_{ij}\right)^{\bot}$
for any $(i,j)\in\mathcal{E}$, therefore $\text{P}_{\mathbb{V}}\left(\boldsymbol{x}^{*}_{i}-\boldsymbol{x}^{*}_{j}\right)=\boldsymbol{0}$,
which implies that $\text{P}_{\mathbb{V}}\boldsymbol{x}^{*}_{i}=\text{P}_{\mathbb{V}}\boldsymbol{x}^{*}_{j}$
for any $(i,j)\in\mathcal{E}$ and the }matrix-weighted network\textcolor{black}{{}
\eqref{eq:overall} achieves subspace consensus on ${\color{violet}{\color{black}\mathbb{V}\subseteq\mathbb{R}^{d}}}$. }

\textcolor{black}{$2\Longleftrightarrow3$: According to Theorem \ref{spanning tree  subspace},
one can immediately obtain the conclusion.}

\textcolor{black}{}
\begin{rem}
Theorem \ref{sufficient and necessary condition} not only provides
a necessary and sufficient condition for subspace consensus on $\mathbb{V}$,
but also reveals that the maximal subspace on which the network can
achieve subspace consensus coincides exactly with the achievable subspace
$\mathbb{V}_{\mathsf{ach}}(\mathcal{T})$ associated with any spanning
tree $\mathcal{T}$. Indeed, if the network achieves subspace consensus
on some subspace $\mathbb{W}\subseteq\mathbb{R}^{d}$, then by Theorem
\ref{sufficient and necessary condition},

\[
\mathbb{W}\subseteq(T^{*}_{ij})^{\perp}=\mathbf{row}(A^{*}_{ij}),
\]
for every edge $(i,j)\in\mathcal{E}$. Consequently,

\[
\mathbb{W}\subseteq\bigcap_{(i,j)\in\mathcal{E}}\mathbf{row}(A^{*}_{ij})=\left(\sum_{(i,j)\in\mathcal{E}}T^{*}_{ij}\right)^{\perp}.
\]
Since $\mathcal{T}$ is a spanning tree of $\mathcal{G}$, its edge
set satisfies $\mathcal{E}(\mathcal{T})\subseteq\mathcal{E}$, and
hence

\[
\bigcap_{(i,j)\in\mathcal{E}}\mathbf{row}(A_{ij})\subseteq\bigcap_{(i,j)\in\mathcal{E}(\mathcal{T})}\mathbf{row}(A_{ij})=\mathbb{V}_{\mathsf{ach}}(\mathcal{T}).
\]
Therefore, $\mathbb{V}_{\mathsf{ach}}(\mathcal{T})$ is the maximal
subspace on which the network can achieve subspace consensus.
\end{rem}
\begin{cor}
\textcolor{black}{\label{collary 3}The }matrix-weighted network\textcolor{black}{{}
\eqref{eq:overall} with limiting constraint network $\mathcal{G}^{*}=(\mathcal{V},\mathcal{E},A^{*})$
achieves subspace consensus on ${\color{violet}{\color{black}\mathbb{V}\subseteq\mathbb{R}^{d}}}$
if and only if $\mathbb{V}\bot T^{*}_{ij}$ for all $(i,j)\in\mathcal{E}$.}
\end{cor}
Theorem \ref{sufficient and necessary condition} provides a criterion
for subspace consensus in terms of the existence of a $\mathbb{V}$-spanning
tree in the limiting constraint network $\mathcal{G}^{*}$. This criterion
applies to general matrix-weighted networks. It is worth noting that
when the original network is itself a tree, there is exactly one simple
path between any two nodes. Consequently, the limiting constraint
network $\mathcal{G}^{*}$ coincides with the original network $\mathcal{G}$
in essence, the limiting candidate subspace $T^{*}_{ij}$ on each
edge is exactly the null space $\mathbf{null}(A_{ij})$ of the original
edge weight. Based on this observation, the characterization of subspace
consensus on tree networks can be expressed solely in terms of the
algebraic properties of the original network. The following Corollary
states this result.

\begin{cor}
\label{cor: tree sufficient and necessary-1}\textcolor{black}{{} The
matrix-weighted network \eqref{eq:overall} on a tree network achieves
subspace consensus on $\mathbb{V}\subseteq\mathbb{R}^{d}$ if and
only if $\mathcal{G}$ is a $\mathbb{V}$-tree.}
\end{cor}
\textcolor{black}{When $\mathbb{V}=\mathbb{R}^{d}$, subspace consensus
reduces to classical consensus. The following theorem characterizes
this case in terms of positive spanning trees in the limiting constraint
network.}
\begin{thm}
\textcolor{black}{\label{sufficient and necessary consition for consensus}}Consider
the \textcolor{black}{matrix-weighted network} \eqref{eq:overall}\textcolor{black}{{}
with limiting constraint network $\mathcal{G}^{*}=(\mathcal{V},\mathcal{E},A^{*})$.
Then the following statements are equivalent: }

\textcolor{black}{1) The matrix-weighted network \eqref{eq:overall}
achieves consensus.}

\textcolor{black}{2) The limiting constraint network $\mathcal{G}^{*}$contains
at least one positive-spanning tree.}

\textcolor{black}{3) Every spanning tree of $\mathcal{G}^{*}$is a
positive-spanning tree.}
\end{thm}
\begin{proof}
\textcolor{black}{Taking $\mathbb{V}=\mathbb{R}^{d}$ in Theorem \ref{sufficient and necessary condition},
subspace consensus reduces to classical consensus, and an $\mathbb{R}^{d}$-spanning
tree is precisely a positive spanning tree. The equivalence of the
three statements then follows directly from Theorem \ref{sufficient and necessary condition}.}
\end{proof}

\section{Conclusion Remarks\label{sec:Conclusion} }

This paper introduced and investigated the subspace consensus problem
for matrix-weighted multi-agent networks. Necessary and/or sufficient
conditions are obtained from both algebraic and topological perspectives.
The limiting constraint network, constructed through an iterative
path-based procedure, was proposed to analyze the steady state of
matrix-weighted networks. Under mild assumptions, we showed that a
matrix-weighted network achieves subspace consensus on ${\color{violet}{\color{black}\mathbb{V}\subseteq\mathbb{R}^{d}}}$
if and only if the associated limiting constraint network contains
a $\mathbb{V}$-spanning tree. This result immediately implies that
a matrix-weighted network achieves consensus if and only if it contains
a positive spanning tree. These results extend the classical consensus
framework and reveal the dimension-specific agreement capabilities
inherent in matrix-weighted networks.

\bibliographystyle{plain}
\bibliography{mybib}

\end{document}